\documentclass{article}
\usepackage[preprint]{spconf}
\usepackage{amsmath,amssymb,dsfont,graphicx,verbatim}
\graphicspath{{figures/}}
\usepackage{amsmath,amssymb,dsfont,verbatim}
\usepackage{import}
\usepackage{subcaption}
\usepackage{cite}
\usepackage{xcolor}
\usepackage{multirow}
\usepackage{booktabs} 
\usepackage{tikz}
\usepackage{siunitx}
\usepackage{balance}
\usepackage[]{algorithm2e}
\usetikzlibrary{mindmap}

\allowdisplaybreaks[0]

\DeclareMathOperator{\cw}{{\scriptstyle\mathcal{W}}}
\DeclareMathOperator{\bcw}{{\boldsymbol{\scriptstyle\mathcal{W}}}}
\usepackage{amsthm}

\DeclareMathOperator{\E}{\mathds{E}}

\DeclareMathOperator{\w}{\boldsymbol{w}}
\DeclareMathOperator{\x}{\boldsymbol{x}}
\DeclareMathOperator{\s}{\boldsymbol{s}}

\usepackage{amsthm}

\theoremstyle{plain}

\newtheorem{assumption}{Assumption}
\newtheorem{theorem}{Theorem}

\newtheorem{lemma}{Lemma}

\title{Competing Adaptive Networks}
\name{Stefan Vlaski and Ali H. Sayed}
\address{School of Engineering, \'{E}cole Polytechnique F\'{e}d\'{e}rale de Lausanne
\thanks{Emails:\{stefan.vlaski, ali.sayed\} @epfl.ch.}}
\begin{document}
\ninept
\maketitle
\pagestyle{empty}
\begin{abstract}
  Adaptive networks have the capability to pursue solutions of \emph{global} stochastic optimization problems by relying only on local interactions within neighborhoods. The diffusion of information through repeated interactions allows for globally optimal behavior, without the need for central coordination. Most existing strategies are developed for \emph{cooperative} learning settings, where the objective of the network is common to all agents. We consider in this work a team setting, where a subset of the agents form a team with a common goal while competing with the remainder of the network. We develop an algorithm for decentralized competition among teams of adaptive agents, analyze its dynamics and present an application in the decentralized training of generative adversarial neural networks.
\end{abstract}
\section{Introduction}\label{sec:introduction}
\subsection{Problem Formulation}
We consider a collection of \( K \) agents, decomposed into two ``teams'' \( \mathcal{N}^{(1)} \) and \( \mathcal{N}^{(2)} \) of size \( K^{(1)} \) and \( K^{(2)} \), respectively. The objective of team \( (1) \) is to coordinate on a common task \( w^{(1)} \), while the objective of team \( (2) \) is to coordinate on another task \( w^{(2)} \), each while optimizing the (potentially) conflicting objectives:
\begin{align}\label{eq:competing_problem_1}
  {w^{(1)}}^o \triangleq \min_{w^{(1)}} J^{(1)}(w^{(1)}; w^{(2)}) \\
  {w^{(2)}}^o \triangleq \min_{w^{(2)}} J^{(2)}(w^{(1)}; w^{(2)})\label{eq:competing_problem_2}
\end{align}
Observe that the objectives are functions of both parameters. It is reasonable to set:
\begin{align}
  J^{(1)}(w^{(1)}; w^{(2)}) = - J^{(2)}(w^{(1)}; w^{(2)}) = J(w^{(1)}; w^{(2)})
\end{align}
in which case we recover the zero-sum game:
\begin{align}\label{eq:zero_sum_problem}
  {w^{(1)}}^o \triangleq \min_{w^{(1)}} J(w^{(1)}; w^{(2)}) \\
  {w^{(2)}}^o \triangleq \max_{w^{(2)}} J(w^{(1)}; w^{(2)})
\end{align}
We will allow for broader choices of \( J^{(1)}(\cdot; \cdot) \) and \( J^{(2)}(\cdot; \cdot) \) for generality. At the team-level, problem~\eqref{eq:competing_problem_1}--\eqref{eq:competing_problem_2} describes a classical two-player game. In the absence of communication constraints within each team, solutions could be pursued through a variety of iterative schemes, such as gradient descent~\cite{Flam02}:
\begin{align}
  w_{i}^{(1)} = w_{i-1}^{(1)} - \mu \nabla_{w^{(1)}} J^{(1)}\left(w_{i-1}^{(1)}; w_{i-1}^{(2)}\right)\label{eq:team_gradient_1} \\
  w_{i}^{(2)} = w_{i-1}^{(2)} - \mu \nabla_{w^{(2)}} J^{(2)}\left(w_{i-1}^{(1)}; w_{i-1}^{(2)}\right)\label{eq:team_gradient_2}
\end{align}
For brevity, we will drop the subscripts \( {w^{(1)}} \) and \( {w^{(2)}} \) in \( \nabla \), with the understanding that gradients of \( J^{(1)} (\cdot; \cdot) \) are taken relative to \( w^{(1)} \) and similarly for \( J^{(2)} (\cdot; \cdot) \). The key difference between the proposed setting, and the types of games considered most often in the literature, is that the team objectives are such that no single agent is able to evaluate \( \nabla J^{(1)}(w_{i-1}^{(1)}; w_{i-1}^{(2)}) \) or \( \nabla J^{(2)}(w_{i-1}^{(1)}; w_{i-1}^{(2)}) \) on its own, and hence collaboration within teams, \textcolor{black}{while competing across teams,} is necessary. In general, each team's objective takes the ``sum-of-costs'' form:
\begin{align}
  J^{(1)}\left(w^{(1)}; w^{(2)}\right) \triangleq \frac{1}{K_1} \sum_{k \in \mathcal{N}^{(1)}} J_k^{(1)}\left(w^{(1)}; w^{(2)}\right) \label{eq:team_one_consensus}\\
  J^{(2)}\left(w^{(1)}; w^{(2)}\right) \triangleq \frac{1}{K_2} \sum_{k \in \mathcal{N}^{(2)}} J_k^{(2)}\left(w^{(1)}; w^{(2)}\right) \label{eq:team_two_consensus}
\end{align}
where each local cost is the average of a loss function:
\begin{align}\label{eq:local_loss}
  J_k^{(t)}\left(w^{(1)}; w^{(2)}\right) \triangleq \mathds{E} Q\left( w^{(1)}; w^{(2)}; \boldsymbol{x}_k^{(t)} \right)
\end{align}
where we are introducting the team variable \( t = \{1, 2 \} \) for brevity. In the above, the variable \(\x_k^{(t)}\) denotes the data that is received at agent \(k\). Returning to~\eqref{eq:team_gradient_1}--\eqref{eq:team_gradient_2}, we find for team \( (t) \):
\begin{align}\label{eq:centralized_recursions}
  w_{i}^{(t)} = w_{i-1}^{(t)} - \mu \sum_{k \in \mathcal{N}^{(t)}} \nabla J_k^{(t)}\left(w_{i-1}^{(1)}; w_{i-1}^{(2)}\right)
\end{align}
We note two drawbacks for this implementation. First, evaluating the full gradient \( \nabla J^{(t)}(w_{i-1}^{(1)}; w_{i-1}^{(2)}) \) requires the central aggregation of all local gradients \( \nabla J^{(t)}(w_{i-1}^{(1)}; w_{i-1}^{(2)}) \) across the team. Second, even in the absence of communication constraints, in light of~\eqref{eq:local_loss}, evaluation of local gradients \( \nabla J^{(t)}(w_{i-1}^{(1)}; w_{i-1}^{(2)}) \) requires knowledge of the data distribution of \( \boldsymbol{x}_k^{(t)} \), which is generally unavailable in practice.

\subsection{Related Works}
Learning problems over graphs are most commonly studied in a \emph{cooperative} setting, where collections of agents coordinate to optimize some global loss function through localized interactions over neighborhoods.
Solutions can be pursued by a variety of decentralized algorithms, including primal~\cite{Nedic09, Chen15transient, Sayed14} and primal-dual~\cite{Shi15, Paolo16, Yuan18, Xin19, Jakovetic20} methods. All of these algorithms yield consensus solutions, where all agents (approximately) converge to a common optimizer of some aggregate loss.

Multi-objective settings, where local objectives differ, and convergence to consensus may not be desired, can be broadly classified into decentralized multi-task learning problems~\cite{Smith17, Nassif20mag} where local objectives do not interfere with each other, and competitive learning problems where choices made by one agent affect the loss of another such as generalized Nash equilibrium problems~\cite{Facchinei07}. This work falls into the latter category. In the important case where the local objective function, in addition to the local action taken by any given agent, depend only on the actions of its neighbors, gradient descent-based approaches~\cite{Flam02} result in naturally decentralized recursions~\cite{ChungKai17}. In the partial-information setting, where local costs depend not only on actions taken within neighborhoods, but also on unobserved actions, schemes based on consensus mechanisms for estimating relevant actions over the graph have been proposed in~\cite{Koshal16, Salehisadaghiani16, Salehisadaghiani19, Tatarenko19}. All these strategies rely on the assumption that, given (an estimate of) the actions of competing agents, each agent is able to evaluate its local objective independently.

In contrast, we consider a setting where subsets of agents form teams with a common objective, in the form of an aggregate loss~\eqref{eq:team_one_consensus}--\eqref{eq:team_two_consensus}, while competing against the remainder of the network. Since the aggregate loss depends on the private data \( \boldsymbol{x}_k^{(t)} \) at each agent \( k \), no single agent is able to evaluate the team objective on its own. More closely related to this setting is the work~\cite{Meng20}, where each cluster designates a representative agent, and interaction among clusters is performed through representative agents via a deterministic gradient-tracking algorithm. In contrast, we present a fully decentralized, and stochastic, algorithm based on the diffusion strategy for decentralized stochastic optimization. While finalizing this manuscript for submission, the work~\cite{Zimmermann21} appeared on arXiv. The authors present a fully decentralized optimization algorithm based on deterministic gradient-tracking, and establish convergence to Nash equilibria under strong-convexity conditions. In contrast, we rely on stochastic gradients, and study the dynamics for general, non-convex loss functions.

\textcolor{black}{\section{Algorithm Development}}
\textcolor{black}{\subsection{Network Model}}
\noindent \textcolor{black}{We denote the set of \( K_1 \) agents belonging to team \( 1\) by \( \mathcal{N}^{(1)} \), and the set of \( K_2 \) agents belonging to team \( (2) \) by \( \mathcal{N}^{(2)} \). The agents in both teams belong to the larger set \( \mathcal{N} \triangleq \mathcal{N}^{(1)} \cup \mathcal{N}^{(2)} \). With each team \( \mathcal{N}^{(1)} \) and \( \mathcal{N}^{(2)} \), we associate graphs with doubly-stochastic adjacency matrices \( A^{(1)} \in \mathds{R}^{K_1 \times K_1} \) and \( A^{(2)} \in \mathds{R}^{K_2 \times K_2} \), respectively. These graphs will be used by the respective teams to coordinate on their local objectives~\eqref{eq:team_one_consensus}--\eqref{eq:team_two_consensus}, and correspond to the blue and red edges in Fig.~\ref{fig:teams}, respectively. However, the coupled nature of~\eqref{eq:team_one_consensus}--\eqref{eq:team_two_consensus}, as we will see, makes it necessary for team \( 1 \) to perform inference about the action \( w^{(2)} \) of team \( 2 \), and vice versa. To this end, we will allow for some minimal flow of information from team \( 2 \) to team \( 1 \), and from team \( 1 \) to team \( 2 \). These links are denoted in grey in Fig.~\ref{fig:teams}, and allow some agents in each team to make inference about the other team's chosen action. We capture these interactions in a second set of adjacency matrices \( A^{(12)} \in \mathds{R}^{K_1 \times K_2} \) and \( A^{(21)} \in \mathds{R}^{K_2 \times K_1} \), where \( A^{(21)} \) captures links from any agent in either team \( \mathcal{N} = \mathcal{N}^{(1)} \cup \mathcal{N}^{(2)} \) to agents belonging to \( \mathcal{N}^{(1)} \), and \( A^{(12)} \) captures linkes from \( \mathcal{N} \) to \( \mathcal{N}^{(2)} \). In other words, \( A^{(21)} \) contains blue and grey links in Fig.~\ref{fig:teams}, while \( A^{(12)} \) contains red and grey links. Successful operation of the algorithm will rely on the diffusion of information through the network. To this end, we introduce the following conditions on the combination matrices.\\
\begin{assumption}[\textbf{Connectivity}]\label{as:connectivity}
  For each \( t \in \{ 1, 2 \} \), the combination matrix \( A^{(t)} \in \mathds{R}^{K_t \times K_t}\) is primitive and doubly-stochastic, ensuring in light of the Perron-Frobenius theorem that:
  \begin{align}
    \lambda_2^{(t)} \triangleq \rho\left( A^{(t)} - \frac{1}{K_t} \mathds{1} \mathds{1}^{\mathsf{T}} \right) < 1
  \end{align}
  The matrix \( A^{(t't)} \in \mathds{R}^{K \times K_t} \) is left-stochastic, i.e., \( \mathds{1}^{\mathsf{T}} A^{(t't)} = \mathds{1}^{\mathsf{T}} \), and furthermore, \( a_{\ell k}^{(t' t)} > 0 \) for at least one \( \ell \in \mathcal{N}^{(t')} \) and \( k \in \mathcal{N}^{(t)} \), ensuring that information flows from the network \( (t') \) to \( (t) \). \hfill \qed
\end{assumption}}

\subsection{Competing Diffusion}
We will be deriving the algorithm from the perspective of team \( (t) \), and denote by \( (t') \) the competing team. Suppose for now that \( w^{(t')} \) is fixed and known to all agents in team \( (t) \). Then, the objective of team~\eqref{eq:team_one_consensus}--\eqref{eq:team_two_consensus} forms a traditional consensus optimization problem. Its solution can be pursued by a number of algorithms for decentralized stochastic optimization. In this work, we will be focusing on the diffusion strategy:
\begin{align}
  \boldsymbol{\phi}_{k, i}^{(t)} =&\: \boldsymbol{w}_{k, i-1}^{(t)} - \mu \widehat{\nabla J}_k^{(t)}\left(\boldsymbol{w}_{k, i-1}^{(t)}; w^{(t')}\right)\label{eq:proposed_update}\\
  \boldsymbol{w}_{k, i}^{(t)} =&\: \sum_{\ell \in \mathcal{N}^{(t)}} a_{\ell k}^{(t)} \boldsymbol{\phi}_{\ell, i}^{(t)}\label{eq:proposed_consensus_1}
\end{align}
Here, the \( a_{\ell k}^{(t)} \) denote the weights of a doubly-stochastic combination matrix \( A^{(t)} \) over the agents of team \( (t) \). While the recursion decentralizes the evaluation of~\eqref{eq:team_gradient_1} over team \( (t) \), it nevertheless requires global access to the quantity \( w^{(t')} \). We remedy this by allowing each agent \( k \in \mathcal{N}^{(t)} \) to maintain an estimate of the action \( w^{(t')} \) by the other team, denoted by \( \boldsymbol{w}_{k, i-1}^{(t')} \), and estimate its value using a direct consensus scheme:
\begin{align}
  \boldsymbol{w}_{k, i}^{(t')} = \sum_{\ell \in \mathcal{N}} a_{\ell k}^{(t't)} \boldsymbol{w}_{\ell, i-1}^{(t')} \label{eq:proposed_consensus_2}
\end{align}
We emphasize a subtle distinction between the consensus steps~\eqref{eq:proposed_consensus_1} and~\eqref{eq:proposed_consensus_2}. While~\eqref{eq:proposed_consensus_1} operates only over the agents \( \mathcal{N}^{(t)} \), step~\eqref{eq:proposed_consensus_2} involves agents that belong to team \( (t') \) and have some connection to team \( (t) \) (see Fig.~\ref{fig:teams}). This allows network \( (t) \) to make inference about the global action \( w^{(t')} \) of team \( (t') \).
\begin{figure}
  \centering
  \includegraphics[width=\columnwidth]{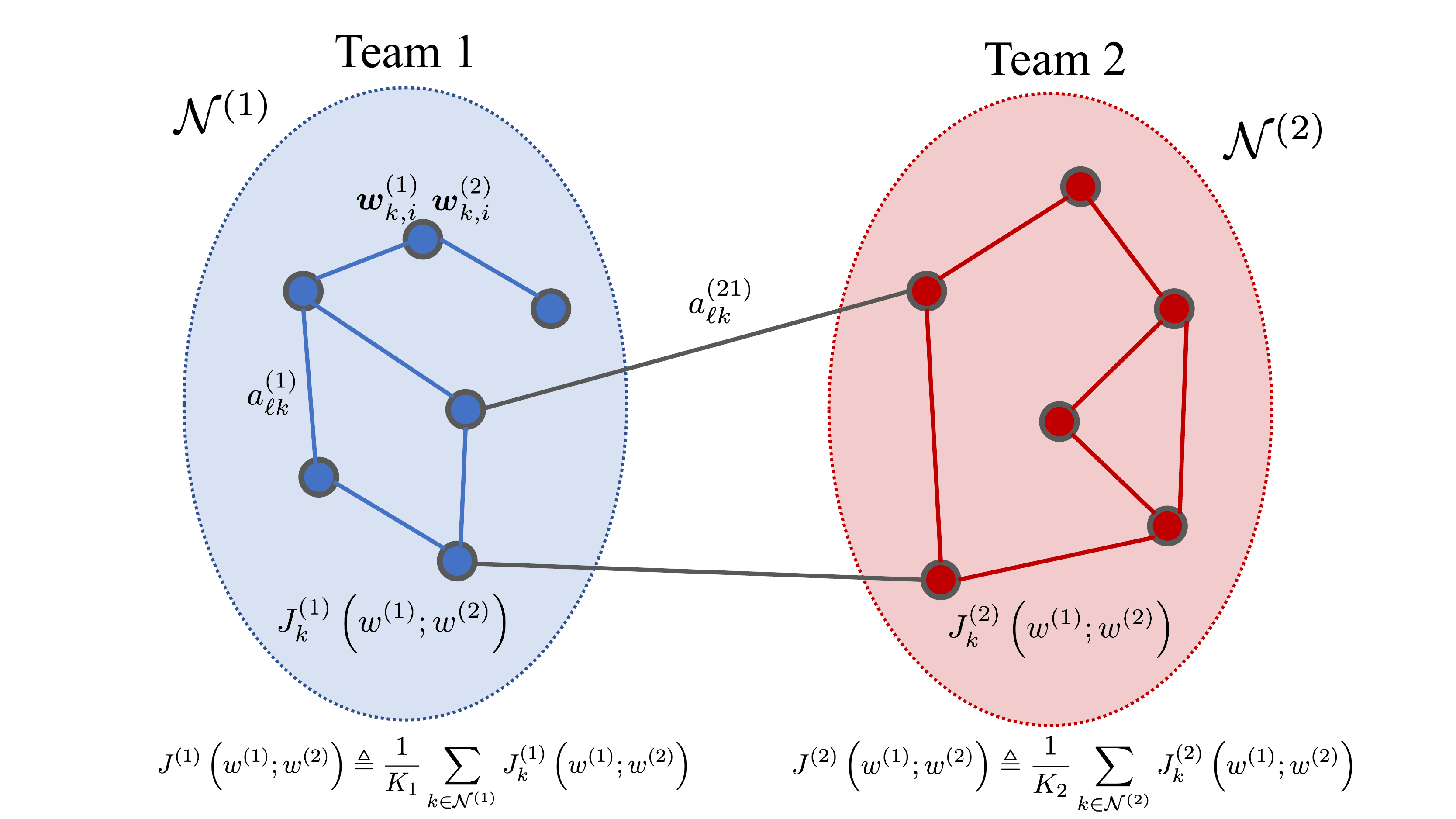}
  \caption{Two competing networks \( \mathcal{N}^{(1)} \) and \( \mathcal{N}^{(2)} \).}\label{fig:teams}
\end{figure}
Replacing \( w^{(t')} \) in~\eqref{eq:proposed_update} by the estimate \( \boldsymbol{w}_{k, i-1}^{(t')} \), we obtain Alg.~\ref{alg:proposed_algorithm}.
\begin{algorithm}
 \For{ \textcolor{black}{\( t \in \{ 1, 2 \} \) and} \(k \in \mathcal{N}^{(t)}\)}{
  \begin{align}
      \hspace{-1cm}\boldsymbol{\phi}_{k, i}^{(t)} =&\: \boldsymbol{w}_{k, i-1}^{(t)} - \mu \widehat{\nabla J}_k^{(t)}\left(\boldsymbol{w}_{k, i-1}^{(t)}; \boldsymbol{w}_{k, i-1}^{(t')}\right)\label{eq:adapt_team_1}\\
    \hspace{-1cm}\boldsymbol{w}_{k, i}^{(t)} =&\: \sum_{\ell \in \mathcal{N}^{(1)}} a_{\ell k}^{(t)} \boldsymbol{\phi}_{\ell, i}^{(1)}\label{eq:combine_1_team_1}\\
    \hspace{-1cm}\boldsymbol{w}_{k, i}^{(t')} =&\: \sum_{\ell \in \mathcal{N}} a_{\ell k}^{(t't)} \boldsymbol{w}_{\ell, i-1}^{(t')}\label{eq:combine_2_team_1}
  \end{align}
 }
 \caption{Diffusion for competing networks.}\label{alg:proposed_algorithm}
\end{algorithm}

\section{Convergence Analysis}
For agents \( k \in \mathcal{N}^{(t)} \) belonging to team \( t \in \{ 1, 2 \} \), we have:
\begin{align}
  \boldsymbol{w}_{k, i}^{(t)} =&\: \sum_{\ell \in \mathcal{N}^{(t)}} a_{\ell k}^{(t)} \left\{ \boldsymbol{w}_{\ell, i-1}^{(t)} - \mu \widehat{\nabla J}_{\ell}^{(t)}\left(\boldsymbol{w}_{\ell, i-1}^{(t)}; \boldsymbol{w}_{\ell, i-1}^{(t')}\right) \right\}\label{eq:summarized_1}
\end{align}
while \( \boldsymbol{w}_{k, i}^{(t')} \) for the other team \( t' \in \{ 2, 1 \} \) is computed according to:
\begin{align}\label{eq:summarized_2}
  \boldsymbol{w}_{k, i}^{(t')} =&\: \sum_{\ell \in \mathcal{N}} a_{\ell k}^{(t't)} \boldsymbol{w}_{\ell, i-1}^{(t')}
\end{align}
For cooperative networks, decentralized recursions of the diffusion type have been shown to rapidly cluster around a particular network centroid, both for convex~\cite{Chen15transient} and nonconvex costs~\cite{Vlaski19nonconvexP1}. In contrast to these works, the gradient appearing in~\eqref{eq:summarized_1} does not solely depend on the local iterates \( \boldsymbol{w}_{\ell, i-1}^{(t)} \), but also on estimates \( \boldsymbol{w}_{\ell, i-1}^{(t')} \), which evolve according to~\eqref{eq:summarized_2} and track \( w^{(t')} \); a quantity which is neither observed by, nor under the control of agents in \( (t) \). We show here that the clustering dynamics continue to hold in a competitive team setting, where agents within a team cluster quickly (at a linear rate determined by the mixing rate of the graph) around a common network centroid, whose evolution tracks that of the centralized recursions~\eqref{eq:centralized_recursions}. To this end, we introduce the centroid vectors:
\begin{align}
  \boldsymbol{w}_{c, i}^{(t)} = \frac{1}{K_t} \sum_{k \in \mathcal{N}^{(t)}} \boldsymbol{w}_{k, i}^{(t)}
\end{align}
Note that team centroids are computed by averaging only the iterates of the given team \( \mathcal{N}^{(t)} \). Then,
\begin{align}
  \boldsymbol{w}_{c, i}^{(t)}=&\: \frac{1}{K_t} \sum_{k \in \mathcal{N}^{(t)}} \sum_{\ell \in \mathcal{N}^{(t)}} a_{\ell k}^{(t)} \left( \boldsymbol{w}_{\ell, i-1}^{(t)} - \mu \widehat{\nabla J}_{\ell}^{(t)}\left(\boldsymbol{w}_{\ell, i-1}^{(t)}; \boldsymbol{w}_{\ell, i-1}^{(t')}\right) \right) \notag \\
  =&\: \frac{1}{K_t} \sum_{\ell \in \mathcal{N}^{(t)}}\sum_{k \in \mathcal{N}^{(t)}} a_{\ell k}^{(t)} \left( \boldsymbol{w}_{\ell, i-1}^{(t)} - \mu \widehat{\nabla J}_{\ell}^{(t)}\left(\boldsymbol{w}_{\ell, i-1}^{(t)}; \boldsymbol{w}_{\ell, i-1}^{(t')}\right) \right) \notag \\
  =&\: \boldsymbol{w}_{c, i-1}^{(t)} - \frac{\mu}{K_t} \sum_{\ell \in \mathcal{N}^{(t)}} \widehat{\nabla J}_{\ell}^{(t)}\left(\boldsymbol{w}_{\ell, i-1}^{(t)}; \boldsymbol{w}_{\ell, i-1}^{(t')}\right)
\end{align}
\begin{assumption}[\textbf{Smoothness}]\label{as:smoothness}
  For each \( k \in \mathcal{N} \) and all \( t \in \{ 1, 2 \} \), the gradient approximation \( \widehat{\nabla J}_k(\cdot, \cdot) \) is Lipschitz in both arguments, namely, for any \( x_1, x_2, y_1, y_2 \in \mathds{R}^{M} \):
  \begin{align}\label{eq:lipschitz_1}
    \|\widehat{\nabla J}_k^{(t)}(x_1; y_1) - \widehat{\nabla J}_k^{(t)}(x_2; y_1)\| \le \delta \|x_1-x_2\| \\
    \|\widehat{\nabla J}_k^{(t)}(x_1; y_1) - \widehat{\nabla J}_k^{(t)}(x_1; y_2)\| \le \delta \|y_1-y_2\|\label{eq:lipschitz_2}
  \end{align}
  Further, the gradients are bounded for all \( x, y \in \mathds{R} \) by:
  \begin{align}
    \|\nabla J_k^{(t)}(x; y)\| \le G
  \end{align}\hfill\qed
\end{assumption}
\begin{assumption}[\textbf{Gradient noise process}]\label{as:gradientnoise}
  For each team \( (t) \) and \( k \in \mathcal{N}^{(t)} \), the gradient noise process is defined as
  \begin{align}
    &\:\s_{k,i}^{(t)}\left(\w_{k,i-1}^{(t)}, \w_{k, i-1}^{(t')}\right)\notag \\
    \triangleq&\: \widehat{\nabla J}_k\left(\w_{k,i-1}^{(t)}, \w_{k, i-1}^{(t')}\right) - \nabla J_k\left(\w_{k,i-1}^{(t)}, \w_{k, i-1}^{(t')}\right)
  \end{align}
  and satisfies
  \begin{subequations}
    \begin{align}
      \E \left\{ \s_{k,i}^{(t)}\left(\w_{k,i-1}^{(t)}, \w_{k, i-1}^{(t')}\right) | \boldsymbol{\mathcal{F}}_{i-1} \right\} &= 0 \label{eq:conditional_zero_mean}\\
      \E \left\{ \|\s_{k,i}^{(t)}\left(\w_{k,i-1}^{(t)}, \w_{k, i-1}^{(t')}\right)\|^2 | \boldsymbol{\mathcal{F}}_{i-1} \right\} &\le {\sigma^2} \label{eq:gradientnoise_fourth}
    \end{align}
  \end{subequations}
  for some non-negative constant \( {\sigma^2} \). \hfill\qed%
\end{assumption}

\begin{lemma}[\textbf{Within-team consensus}]\label{LEM:WITHIN-TEAM_CONSENSUS}
  All iterates \( \w_{k, i}^{(t)} \) for agents \( k \in \mathcal{N}^{(t)} \) in team \( (t) \) cluster around the team centroid \( \w_{c, i}^{(t)} \) after sufficient iterations \( i^o \), i.e.,
  \begin{align}
    \mathds{E} {\|\w_{k, i}^{(t)} - \w_{c, i}^{t}\|}^2 \le \mu^2 \frac{2{\lambda_2^{(t)}}^2}{1 - {\lambda_2^{(t)}}} K_t \left( \frac{G^2}{1 - {\lambda_2^{(t)}}} + \sigma^2 \right)
  \end{align}
  for
  \begin{align}
    i \ge i^o = O(\log(\mu)) = o(\mu^{-1})
  \end{align}
\end{lemma}
\begin{proof}
  Omitted due to space limitations.
\end{proof}
Lemma~\ref{LEM:WITHIN-TEAM_CONSENSUS} establishes that agents within each team are able to coordinate their parameters \( \w_{k, i}^{(t)} \) to cluster around a common centroid \( \w_{c, i}^{(t)} \) in the mean-square error sense. We now investigate how well agents are able to estimate \( \w_{k, i}^{(t')} \), i.e., the parameters of the competing network. To this end, we exploit the asymmetric flow of information in recursion~\eqref{eq:summarized_2}. In particular, we can decompose for \( k \in \mathcal{N}^{(t)} \):
\begin{align}
  \boldsymbol{w}_{k, i}^{(t')} =&\: \sum_{\ell \in \mathcal{N}} a_{\ell k}^{(t't)} \boldsymbol{w}_{\ell, i-1}^{(t')} \notag \\
  =&\: \sum_{\ell \in \mathcal{N}^{(t)}} a_{\ell k}^{(t't)} \boldsymbol{w}_{\ell, i-1}^{(t')} + \sum_{\ell \in \mathcal{N}^{(t')}} a_{\ell k}^{(t't)} \boldsymbol{w}_{\ell, i-1}^{(t')}
\end{align}
Note that this recursion describes the evolution of \( \boldsymbol{w}_{k, i}^{(t')} \) for \( k \in \mathcal{N}^{(t)} \) only. The estimates \( \boldsymbol{w}_{\ell, i-1}^{(t')} \) for \( \ell \in \mathcal{N}^{(t')} \) evolve independently, and in particular, following Lemma~\ref{LEM:WITHIN-TEAM_CONSENSUS}, track \( \boldsymbol{w}_{c, i-1}^{(t')} \) after sufficient iterations. Such an asymmetric structure is reminiscent to the learning dynamics observed over weakly-connected directed networks, encountered in~\cite{Ying16weakly, Salami17weakly}. Applying these insights to the competitive team setting, we obtain the following lemma.
\begin{lemma}[\textbf{Cross-team learning}]\label{LEM:CROSS-TEAM_LEARNING}
  The estimates for the competing model parameters \( \w_{k, i}^{(t')} \) maintained by agents \( k \in \mathcal{N}^{(t)} \) in team \( (t) \) cluster around the centroid of the competing team \( \w_{c, i}^{(t')} \) after sufficient iterations \( i^o \), i.e.,
  \begin{align}
    \mathds{E} {\|\w_{k, i}^{(t')} - \w_{c, i}^{(t')}\|}^2 \le O(\mu^2)
  \end{align}
  for
  \begin{align}
    i \ge i^o = O(\log(\mu)) = o(\mu^{-1})
  \end{align}
\end{lemma}
\begin{proof}
  Omitted due to space limitations.
\end{proof}
Lemmas~\ref{LEM:WITHIN-TEAM_CONSENSUS} and~\ref{LEM:CROSS-TEAM_LEARNING}, when taken together ensure that team \( (t) \) is able to coordinate on a common model \( \w_{c, i}^{(t)} \), and estimate the competing model \( \w_{c, i}^{(t')} \) with high accuracy for sufficiently small step-sizes (namely within \( O(\mu^2) \)) in the mean-square sense. We combine these to obtain a description of the learning dynamics of the competitive diffusion strategy.
\begin{theorem}[\textbf{Learning dynamics of competitive diffusion}]
  Iterates generated by the competitive diffusion scheme in Algorithm~\ref{alg:proposed_algorithm} approximately follow the centralized batch strategy. Specifically, for \( k \in \mathcal{N}^{(t)} \) and \( i \in \{ 1, 2 \} \), we have for \( i \ge i^o \):
  \begin{align}
    \w_{k, i}^{(t)} = \w_{k, i-1}^{(t)} - \frac{\mu}{K_t} \sum_{k \in \mathcal{N}^{(t)}} \widehat{\nabla J}_k^{(t)}\left(\w_{k, i-1}^{(1)}; \w_{k, i-1}^{(2)}\right) - \mu \boldsymbol{d}_{k, i}
  \end{align}
  where \( i^o = o(\mu^{-1}) \) and:
  \begin{align}
    \E\|\boldsymbol{d}_{k, i}\|^2 \le O(\mu^2)
  \end{align}
\end{theorem}
\begin{proof}
  The result follows directly from Lemma~\ref{LEM:WITHIN-TEAM_CONSENSUS} and~\ref{LEM:CROSS-TEAM_LEARNING} along with the Lipschitz conditions~\eqref{eq:lipschitz_1}--\eqref{eq:lipschitz_2}.
\end{proof}

\section{Numerical Results}
We illustrate how the competitive setting~\eqref{eq:team_one_consensus}--\eqref{eq:team_two_consensus} can be applied to train generative adversarial neural networks (GANs) in a decentralized manner. The objective of GANs is to learn a ``generator'' mapping \( g(w^{(1)}; \boldsymbol{z}) \) from some random noise variable \( \boldsymbol{z} \in \mathds{R}^{M_n} \) to a feature space \( \mathds{R}^{M_f} \), such that objects \( g(w^{(1)}; \boldsymbol{z}) \) generated from pure noise are indistinguishable (in some sense) from features \( \boldsymbol{h} \in \mathds{R}^{M_f} \) following an unkown distribution~\cite{Goodfellow14}. This is accomplished by simultaneously training a ``discriminator'' \( d(w^{(2)}; \widehat{\boldsymbol{h}}) \) to determine whether \( \widehat{\boldsymbol{h}} \) was sampled from \( \boldsymbol{h} \), or from \( g(w^{(1)}; \boldsymbol{z}) \). Specifically, we can let:
\begin{align}
  J_k^{(1)}(w^{(1)}, w^{(2)}) \triangleq&\: \mathds{E}_{\boldsymbol{z}} \log(1 - d(g(\boldsymbol{z})) + \mathds{E}_{\boldsymbol{h}_k} \log(d(\boldsymbol{h}_k))\\
  J_k^{(1)}(w^{(1)}, w^{(2)}) \triangleq&\: - J^{(1)}(w^{(1)}, w^{(2)})
\end{align}
which fits into the framework considered in this work. We illustrate performance on simple, fully-connected, feedforward neural networks trained using Alg.~\ref{alg:proposed_algorithm} and depict loss evolution in Fig.~\ref{fig:losses} and the evolution of generated images in~\ref{fig:digits}.
\begin{figure}[htb]
	\centering
	\includegraphics[width=.95\columnwidth]{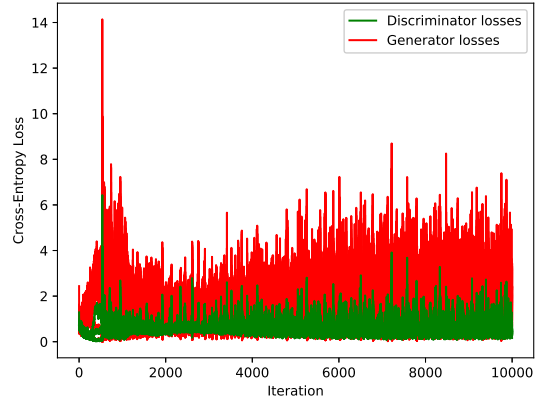}
\caption{Generator and discriminator losses for a collection of \( K = 16 \) competing agents.}\label{fig:losses}
\end{figure}
\begin{figure}[htb]
	\centering
	\includegraphics[width=.95\columnwidth]{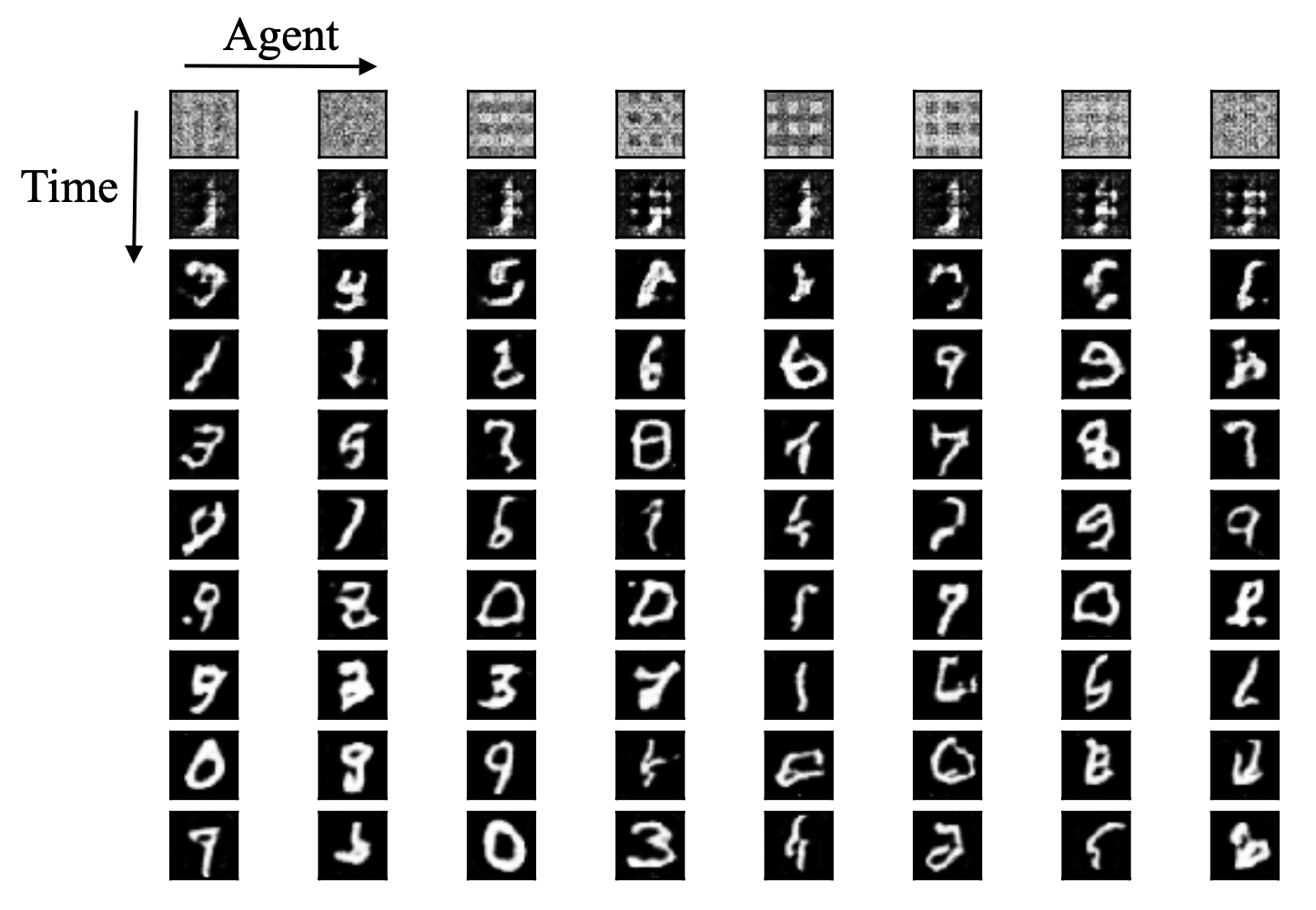}
\caption{Digits generated by the different competing generators (from left to right) as time learning progresses (top to bottom).}\label{fig:digits}
\end{figure}

\bibliographystyle{IEEEbib}
\bibliography{main}

\end{document}